\documentclass[letterpaper, 10 pt,conference]{ieeeconf}  

\IEEEoverridecommandlockouts                 
\usepackage{booktabs}
\usepackage{multirow}
\usepackage[hidelinks]{hyperref}

\usepackage{mathrsfs,amsmath,amssymb,stmaryrd,amsfonts,mathtools}
\usepackage{multirow}
\usepackage[linesnumbered,ruled,vlined]{algorithm2e}
\usepackage{graphicx,algpseudocode}
\usepackage{color}
\usepackage{footnote}
\usepackage{color}
\usepackage{makecell}

\def\defas{\ensuremath{\mathrel{:=}}}

\def\Set#1#2{\ensuremath{
\left\{#1\,\middle|\,#2\right\}
}}

\def\dist{\mathrm{dist}}

\def\dh{\mathbf{d}_{\mathrm{H}}^{\mathrm{s}}}

\def\Comp{\mathcal{K}(\mathbb{R}^{n})}

\def\CompDOA{\mathcal{K}_{{\mathcal{D}}_{\mathcal{A}}}(\mathbb{R}^{n})}

\def\norm#1{\|  #1 \| }

\DeclarePairedDelimiter\abs{\lvert}{\rvert}

\DeclareMathDelimiter{\orbrack}{\mathopen}{operators}{"5D}{largesymbols}{"03}
\DeclareMathDelimiter{\clbrack}{\mathclose}{operators}{"5B}{largesymbols}{"02}
\def\intcc#1{\ensuremath{[#1]}}
\def\intoo#1{\ensuremath{\orbrack#1\clbrack}}
\def\intoc#1{\ensuremath{\orbrack#1]}}
\def\intco#1{\ensuremath{[#1\clbrack}}

\newtheorem{theorem}{Theorem}
\newtheorem{lemma}[theorem]{Lemma}

\newtheorem{assumption}{Assumption}

\title{\LARGE \bf
Set-Based Value Function Characterization and Neural Approximation of Stabilization Domains for Input-Constrained Discrete-Time Systems
}

\author{Mohamed Serry, S. Sivaranjani,  and Jun Liu
\thanks{M.~S. and J.~L. are with the Department of Applied Mathematics, University of Waterloo, Waterloo, Ontario, Canada.  
The work of M.~S. and J.~L. was funded in part by the Natural Sciences and Engineering Research Council of Canada  (e-mail: \{mserry,~j.liu\}@uwaterloo.ca). M.~S. is also with the Department of Mechanical and Mechatronics Engineering at the University of Waterloo.}
 \thanks{S.~S. is with the School of Industrial Engineering, Purdue University, 15 N Grant Street, West Lafayette, 47907, Indiana. S.~S. was supported in part by the Air Force Office of Scientific Research grant FA9550-23-1-0492.}}%

\begin{document}

\maketitle
\thispagestyle{empty}
\pagestyle{empty}

\begin{abstract}
Analyzing nonlinear systems with stabilizable controlled invariant sets (CISs) requires accurate estimation of their domains of stabilization (DOS) together with associated stabilizing controllers. Despite extensive research, estimating DOSs for general nonlinear systems remains challenging due to fundamental theoretical and computational limitations.  In this paper, we propose a novel framework for estimating DOSs for controlled input-constrained  discrete-time systems. The DOS is characterized via newly introduced value functions defined on metric spaces of compact sets. We establish the fundamental properties of these value functions and derive the associated Bellman-type (Zubov-type) functional equations. Building on this characterization, we develop a physics-informed neural network (NN) framework that learns the value functions by embedding the derived functional equations directly into the training process. The proposed methodology is demonstrated through two numerical examples, illustrating its ability to accurately estimate DOSs and synthesize stabilizing controllers from the learned value functions.
\end{abstract}

\section{Introduction}
A fundamental problem in control theory is the design of stabilizing controllers for nonlinear systems where the goal is to construct a feedback law, that is, a mapping from the state to the control input, such that the resulting closed-loop system is asymptotically stable with respect to a desired equilibrium point or a target set of interest. Closely related to this objective is the characterization of the set of initial conditions that can be driven asymptotically to the target set, commonly referred to as the domain of stabilization (DOS). An overview of DOS estimation for continuous-time systems can be found in, e.g., \cite{grune2000computing, liu2025formal}.

In this work, we study DOSs for controlled discrete-time nonlinear systems. The stabilization of such systems has been extensively investigated, with a wide range of approaches including linearization-based methods \cite{weiss1972controllability}, control Lyapunov functions \cite{kellett2004discrete, wu2024control}, sliding mode control \cite{koshkouei2000sliding, jakubczyk1987feedback}, feedback linearization \cite{aranda1996linearization}, and model predictive control \cite{grune2016nonlinear}. While these methods have demonstrated effectiveness across a broad range of applications, a fundamental challenge remains: the resulting controllers do not, in general, guarantee that all initial states within the DOS are driven asymptotically to the target set or equilibrium point. In particular, stabilization guarantees are often local, and the domain of attraction of the closed-loop system does not necessarily coincide with the DOS, especially in the presence of input constraints. As a result, the domain of attraction of the implemented controller may be strictly smaller than the true DOS.

An alternative and well-established framework for studying stabilization is based on value functions. In this approach, a running cost is introduced to penalize deviations from the desired behavior, and the associated value function is defined as the optimal accumulated cost. This function satisfies a Bellman-type equation \cite{grune2007approximately}, which can, in principle, be solved via dynamic programming. This framework provides a principled way to construct stabilizing controllers and to characterize the DOS exactly as sublevel sets of the value function. However, its practical implementation typically requires state-space discretization, leading to exponential growth in computational complexity.

Recently, learning-based approaches have been proposed for synthesizing neural controllers and estimating domains of attraction of the resulting closed-loop systems \cite{dai2021lyapunov, wu2023neural, chang2019neural,zhou2022neural}. These methods typically enforce Lyapunov-type decrease conditions during training. While effective in practice, they generally do not guarantee maximality of the resulting DOS estimates. Certification of learned neural Lyapunov functions and controllers is commonly performed using formal verification tools, such as those based on interval arithmetic or mixed-integer programming \cite{wu2023neural,shi2024certified,gao2013dreal}. Although neural network verification remains computationally demanding, recent advances in linear bound propagation and branch-and-bound techniques have significantly improved scalability \cite{bunel2020branch, xu2021fast, wang2021beta, ferrari2022complete, zhou2024scalable}. These developments have motivated ongoing efforts toward obtaining certifiable neural Lyapunov functions and controllers valid over large regions of the state space.

A natural direction is to combine the theoretical guarantees of value-function-based methods with the scalability of learning-based approaches, with the goal of learning maximal Lyapunov functions. Such ideas have been explored for uncontrolled systems (see, e.g., \cite{serry2025safe, serry2026safe}). However, extending these approaches to controlled systems introduces several challenges. First, evaluating value functions requires solving infinite-horizon optimal control problems, which are generally intractable; even finite-horizon approximations lead to large-scale nonconvex optimization problems. Second, the Bellman equations involve infimum operators over the control inputs, which complicates their direct incorporation into physics-informed learning frameworks, i.e., learning frameworks that embed governing equations into the loss functions. Physics-informed learning is particularly attractive in control applications, as it incorporates system dynamics and stability conditions directly into the training process, thereby improving the approximation quality of the learned neural functions compared to purely data-driven approaches 
\cite{liu2025physics,liu2025formally}.

Motivated by these challenges, we develop a novel framework for estimating DOSs associated with prescribed  CISs for discrete-time nonlinear systems, together with corresponding neural network controllers. The proposed approach integrates two key components:
(i) set-based value functions that characterize the DOS,
and (ii) neural network approximations to represent these functions.

Specifically, we introduce value functions defined on metric spaces of compact sets that characterize the DOS while enabling tractable evaluation via reachable-set approximations. These value functions satisfy Bellman-type functional equations without explicit infimum operators, allowing practical evaluation of residual errors through finite-dimensional embeddings of set representations. This structure facilitates physics-informed learning of the value functions. 

The remainder of the paper is organized as follows. Section~\ref{sec:Preliminaries} introduces the notation and preliminary results. Section~\ref{sec:ProblemFormulation} presents the problem formulation. Section~\ref{sec:ValueFunctions} introduces the proposed value functions and establishes their theoretical properties. Section~\ref{sec:VFApproximation} presents the learning procedure for approximating the value functions. Section~\ref{sec:NumericalExamples} provides numerical examples illustrating the potential of the proposed framework for DOS estimation and control synthesis. Finally, Section~\ref{sec:Conclusion} concludes the paper.

\section{Notation and preliminaries}
\label{sec:Preliminaries}
Let $\mathbb{R}$, $\mathbb{R}_+$, $\mathbb{Z}$, and $\mathbb{Z}_{+}$ denote the sets of real numbers, nonnegative real numbers, integers, and nonnegative integers, respectively, and define $\mathbb{N} := \mathbb{Z}_{+} \setminus \{0\}$. 
For $a,b \in \mathbb{R}$ with $a \le b$, let $\intcc{a,b}$, $\intoo{a,b}$, $\intco{a,b}$, and $\intoc{a,b}$ denote the closed, open, and half-open intervals in $\mathbb{R}$ with endpoints $a$ and $b$. 
Their discrete counterparts are denoted by $\intcc{a;b}$, $\intoo{a;b}$, $\intco{a;b}$, and $\intoc{a;b}$, respectively.  Let $\norm{\cdot}$ denote any vector norm on $\mathbb{R}^{n}$ and $\mathbb{B}_{n}$ be the $n$-dimensional closed unit ball induced by $\norm{\cdot}$.   The Minkowski sum of $X,Y\subseteq \mathbb{R}^{n}$ is defined as
$
X+Y\defas \Set{x+y}{x\in X,~y\in Y}.
$
The class of nonempty compact subsets of $X\subseteq \mathbb{R}^{n}$ is denoted by $\mathcal{K}(X)$. Given a point $x\in \mathbb{R}^{n}$ and a set $\Omega\subseteq \mathbb{R}^{n}$, define the distance from $x$ to $\Omega$ as $\dist(x,\Omega)\defas \inf_{y\in \Omega}\norm{x-y}.
$
The distance function $\dist$ induces the (symmetric) Hausdorff distance
\[
\dh(\Omega,\Gamma)
\defas
\max\left\{
\sup_{x\in\Omega}\dist(x,\Gamma),
\sup_{x\in\Gamma}\dist(x,\Omega)
\right\},
\]
$\Omega,\Gamma\subseteq\mathbb{R}^n$; see, for example, \cite{serry2026safe,serry2021overapproximating}.
Note that, for $X\subseteq \mathbb{R}^{n}$, $\mathcal{K}(X)$, equipped with the Hausdorff distance  $\dh$, is a metric space. 
Given $f\colon X\rightarrow Y$ and  $P\subseteq X$, the image  of $f$ on $P$ is defined as  $f(P)\defas\Set{f(x)}{x\in P}$.  Given $f\colon X \rightarrow X$ and  $x\in X$,  $f^{0}(x)\defas x$, and for  $M\in \mathbb{N}$, we define $f^{M}(x)$ recursively as follows:   $f^{k}(x)=f(f^{k-1}(x)),~k\in \intcc{1;M}.
$
Given two sets $X$ and $Y$, ${Y}^{X}$ denotes the set of all maps $f\colon X\rightarrow Y$.

\section{Problem formulation}\label{sec:ProblemFormulation}

Consider the discrete-time   system 
\begin{equation}\label{eq:System}
x_{k+1}=f(x_{k},u_{k}),~k\in \mathbb{Z}_{+},
\end{equation}
where  $x_{k}\in \mathbb{R}^{n}$ is the system state, $u_{k}\in \mathbb{R}^{m}$ is the control input, satisfying $
    u_{k}\in U,~k\in \mathbb{Z}_{+},
    $
    where   $U\in \mathcal{K}(\mathbb{R}^{m})$ is a known input set, and  $f\colon \mathbb{R}^{n}\times \mathbb{R}^{m}\rightarrow \mathbb{R}^{n}$ is a continuous function over $\mathbb{R}^{n}\times \mathbb{R}^{m}$ specifying the dynamics of the system. Using the function $f$, we define the map  $F\colon \mathcal{K}(\mathbb{R}^{n})\rightarrow \mathcal{K}(\mathbb{R}^{n})$ as follows:
$$
F(X)\defas f(X,U)=\bigcup_{(x,u)\in X\times U}f(x,u),~X\in \mathcal{K}(\mathbb{R}^{n}).
$$
Note that
    $F$ is continuous w.r.t. Hausdorff distance $\dh$  \cite[Lemma~3]{serry2026safe}.
    
    A trajectory $\varphi_{x}^{\pi}\colon\mathbb{Z}_{+}\rightarrow \mathbb{R}^{n}$ of system \eqref{eq:System}, induced by an initial condition $x\in \mathbb{R}^{n}$ and an input signal $\pi \colon \mathbb{Z}_{+}\rightarrow U$, satisfies
 $$
 \varphi_{x}^{\pi}(0)=x,~
    \varphi_{x}^{\pi}(k+1)=f(\varphi_{x}^{\pi}(k),\pi(k)),~k\in \mathbb{Z}_{+}.
    $$
The  reachable set of $X\subseteq \mathbb{R}^{n}$ at time $k\in \mathbb{Z}_{+}$, under the dynamics of \eqref{eq:System}, is defined as:
    \begin{equation}\label{eq:ReachSet}
\mathcal{R}(X,k)\defas\{y\in \mathbb{R}^{n}| \exists (x,\pi) \in X\times U^{\mathbb{Z}_{+}} ~\text{s.t.}~ y=\varphi_{x}^{\pi}(k)\}.
    \end{equation}  
The following properties of reachable sets are essential in our analysis in this work:
\begin{lemma}[Lemma~4 in~\cite{serry2026safe}]\label{lem:SemiGroup} For all $(X,k)\in \mathcal{K}(\mathbb{R}^{n})\times\mathbb{N},$
    $$
    \mathcal{R}(X,k)=F(\mathcal{R}(X,k-1))=\mathcal{R}(F(X),k-1)=F^{k}(X).
    $$
Additionally, for $X\in \Comp$, $\mathcal{R}(X,k)\in \Comp$  for all $k\in \mathbb{Z}_{+}$. Moreover, for fixed $k\in \mathbb{Z}_{+}$, $\mathcal{R}(\cdot,k)\colon \Comp\rightarrow \Comp$ is continuous w.r.t. the Hausdorff distance $\dh$.
\end{lemma} 
Let 
$\mathcal{A} \in \mathcal{K}(\mathbb{R}^{n})$ be a controlled invariant set (CIS), that is,
$$
\forall x\in \mathcal{A} ~\exists u\in U~\text{s.t.}~ f(x,u)\in \mathcal{A}.
$$

We impose the following stabilizability assumption, which covers both exponential and polynomial stabilizability as special cases.

\begin{assumption}[Local $\ell_{p}$-stabilizability]\label{Assumptions}
The set $\mathcal{A}$ is locally $\ell_p$-stabilizable. That is, there exist constants $r>0$, $M\ge 1$, $p>0$, and a function 
$\lambda:[0,r]\times \mathbb{Z}_{+}\to \mathbb{R}_{+}$ such that:
\begin{enumerate}
\item For each $k\in \mathbb{Z}_{+}$, $s\mapsto \lambda(s,k)$ is continuous, nondecreasing on $[0,r]$, and $\lambda(0,k)=0$. For each $s\in [0,r]$, $k\mapsto \lambda(s,k)$ is nonincreasing and $\lambda(s,0)\le s$;

\item $\sum_{k=0}^{\infty} \lambda(r,k)^p < \infty$;

\item For every $x\in \mathbb{R}^n$ with $\dist(x,\mathcal{A})\le r$, there exists $\pi\in U^{\mathbb{Z}_{+}}$ such that
\[
\dist\!\bigl(\varphi_x^\pi(k),\mathcal{A}\bigr)
\le M\,\lambda\!\bigl(\dist(x,\mathcal{A}),k\bigr),
\quad \forall k\in \mathbb{Z}_{+}.
\]
\end{enumerate}
\end{assumption}

The DOS w.r.t. $\mathcal{A}$,
$\mathcal{D}_{\mathcal{A}}\subseteq \mathbb{R}^{n}$, is defined as
$$
\mathcal{D}_{\mathcal{A}}\defas \left\{x\in \mathbb{R}^{n} \middle\vert \begin{array}{c}\exists \pi\in U^{\mathbb{Z}_{+}}~\text{s.t.}\\
      \lim_{k\rightarrow \infty} \dist( \varphi_{x}^{\pi}(k),\mathcal{A})=0\end{array}
\right\}.
$$

 Our objective is to compute an accurate approximation of $\mathcal{D}_{\mathcal{A}}$. Under Assumption~\ref{Assumptions}, the set
$\mathcal{D}_{\mathcal A}$ is open. Moreover, local
$\ell_p$-stabilizability implies the equivalent characterization: 
\[
\mathcal{D}_{\mathcal{A}}
=
\Set{x\in \mathbb{R}^{n}}{\lim_{k\to\infty}\inf_{y\in \mathcal{R}(\{x\},k)} \dist(y,\mathcal{A}) = 0}.
\]

\section{Value Functions}
\label{sec:ValueFunctions}

In this section, we introduce novel value functions defined on metric spaces of compact sets that characterize the DOS.

To this end, let $\alpha:\mathbb{R}^n \to \mathbb{R}_+$ be a continuous function satisfying
\begin{equation}\label{eq:AlphaBounds}
\underline{\alpha}\,\dist(x,\mathcal{A})^{\bar p}
\;\le\;
\alpha(x)
\;\le\;
\overline{\alpha}\,\dist(x,\mathcal{A})^{\bar p},
\quad x\in \mathbb{R}^n,
\end{equation}
for some constants $\underline{\alpha},\overline{\alpha}>0$ and $\bar p\ge p$. Such a function always exists; for instance, one may choose $\alpha(x)=\dist(x,\mathcal{A})^{p}$.

We define the value functions $\mathcal{V},\mathcal{W}:\Comp \to \mathbb{R}_+\cup\{\infty\}$ by
\begin{align}
\mathcal{V}(X) &\defas \sum_{k=0}^{\infty} \Psi(\mathcal{R}(X,k)), \label{eq:Lyapunov1}\\
\mathcal{W}(X) &\defas 1 - \exp\big(-\mathcal{V}(X)\big), \label{eq:Lyapunov2}
\end{align}
where $\Psi:\Comp \to \mathbb{R}_+$ is given by
\begin{equation}\label{eq:Psi}
\Psi(X)\defas \inf_{y\in X} \alpha(y), \quad X\in \Comp,
\end{equation}
with the convention $\exp(-\infty)=0$.

The value functions $\mathcal{V}$ and $\mathcal{W}$ provide a characterization of the DOS, as shown below.
\begin{theorem}\label{thm:DOASublevel}
Define
$$
\mathbb{V}_{\infty} \defas \{x\in \mathbb{R}^{n} \mid \mathcal{V}(\{x\})<\infty\}, 
$$
and
$$
\mathbb{W}_{1} \defas \{x\in \mathbb{R}^{n} \mid \mathcal{W}(\{x\})<1\}.
$$
Then, $\mathbb{V}_{\infty} = \mathbb{W}_{1} = \mathcal{D}_{\mathcal{A}}.
$
\end{theorem}

\begin{proof}
Since $\mathcal{W}(\{x\}) < 1$ if and only if $\mathcal{V}(\{x\}) < \infty$ for all $x \in \mathbb{R}^n$, it suffices to prove that
$
\mathbb{V}_{\infty} = \mathcal{D}_{\mathcal{A}}$. Recall the definitions of $M$, $r$, and $\lambda$ from Assumption \ref{Assumptions}.

\noindent\textbf{($\mathbb{V}_{\infty} \subseteq \mathcal{D}_{\mathcal{A}}$)}
Let $x \in \mathbb{V}_{\infty}$. Then the series
\[
\sum_{k=0}^{\infty} \Psi\big(\mathcal{R}(\{x\},k)\big)
\]
converges. Hence, its terms converge to zero. By the definition of
$\Psi$ in \eqref{eq:Psi},
\[
\lim_{k\to\infty}
\inf_{y \in \mathcal{R}(\{x\},k)} \alpha(y)=0.
\]
Using the lower bound on $\alpha$ in \eqref{eq:AlphaBounds} and the
squeeze theorem, it follows that
\[
\lim_{k\to\infty}
\inf_{y \in \mathcal{R}(\{x\},k)}
\dist(y,\mathcal{A})^{\bar p}=0.
\]
Therefore, there exists $K\in\mathbb{Z}_{+}$ such that
\[
\inf_{y\in\mathcal{R}(\{x\},K)}
\dist(y,\mathcal{A})^{\bar{p}}\leq r^{\bar{p}}.
\]
Since $\mathcal{R}(\{x\},K)$ is compact by Lemma~\ref{lem:SemiGroup},
this infimum is attained. Thus, there exists
$y\in\mathcal{R}(\{x\},K)$ such that
\[
\dist(y,\mathcal{A})\leq r.
\]
Equivalently, and using the definition of the reachable sets in \eqref{eq:ReachSet}, there exist an input signal
$\pi_1\in U^{\mathbb{Z}_{+}}$ and a time $K\in\mathbb{Z}_{+}$ such that
\[
\dist\bigl(\varphi_x^{\pi_1}(K),\mathcal{A}\bigr)\leq r.
\]

By $\ell_p$-stabilizability (Assumption~\ref{Assumptions}), there exists an input signal $\pi_2 \in U^{\mathbb{Z}_{+}}$ such that $$
\dist\big(\varphi_{\varphi_{x}^{\pi_1}(K)}^{\pi_2}(k),\mathcal{A}\big)
\le M\,\lambda\big(\dist(\varphi_{x}^{\pi_1}(K),\mathcal{A}),k\big), ~k \in \mathbb{Z}_{+}.
$$
Define $\pi \in U^{\mathbb{Z}_{+}}$ by concatenation:
\begin{equation}\label{eq:Conctenation}
\pi(k) =
\begin{cases}
\pi_1(k), & k < K,\\
\pi_2(k-K), & k \ge K.
\end{cases}
\end{equation}
Then $\dist(\varphi_{x}^{\pi}(k),\mathcal{A}) \to 0$ as $k \to \infty$, and thus $x \in \mathcal{D}_{\mathcal{A}}$.

\medskip
\noindent\textbf{($\mathcal{D}_{\mathcal{A}} \subseteq \mathbb{V}_{\infty}$)}
Let $x \in \mathcal{D}_{\mathcal{A}}$. Then there exists $\pi_1 \in U^{\mathbb{Z}_{+}}$ such that
$ \lim_{k\rightarrow \infty}
\dist(\varphi_{x}^{\pi_1}(k),\mathcal{A}) = 0$. Hence, there exists $K \in \mathbb{Z}_{+}$ such that
$
\dist(\varphi_{x}^{\pi_1}(K),\mathcal{A}) \le r.
$
By $\ell_p$-stabilizability, there exists $\pi_2 \in U^{\mathbb{Z}_{+}}$ such that $\dist\big(\varphi_{\varphi_{x}^{\pi_1}(K)}^{\pi_2}(k),\mathcal{A}\big)
\le M\,\lambda\big(\dist(\varphi_{x}^{\pi_1}(K),\mathcal{A}),k\big)$,  $k \in \mathbb{Z}_{+}$. Define $\pi$ as in \eqref{eq:Conctenation}. Then, using the definition of $\Psi$ in \eqref{eq:Psi} and the bounds on $\alpha$ in \eqref{eq:AlphaBounds}, we obtain
\begin{align*}
\sum_{k=K}^{\infty} \Psi\big(\mathcal{R}(\{x\},k)\big)
&\le \sum_{k=K}^{\infty} \alpha(\varphi_{x}^{\pi}(k))\\
&\le \overline{\alpha} \sum_{k=K}^{\infty} \dist(\varphi_{x}^{\pi}(k),\mathcal{A})^{\bar p}\\
&\le \overline{\alpha} M^{\bar p} \sum_{k=0}^{\infty} \lambda\big(\dist(\varphi_{x}^{\pi_1}(K),\mathcal{A}),k\big)^{\bar p}\\
&< \infty.
\end{align*}
 Therefore, $\mathcal{V}(\{x\}) < \infty$, and hence $x \in \mathbb{V}_{\infty}$. Combining both inclusions completes the proof.
\end{proof}

\subsection{Properties of the Value Functions}
\label{sec:VFProperties}

In this section, we establish basic properties of the value functions $\mathcal{V}$ and $\mathcal{W}$. 

To this end, we define:
\[
\CompDOA \defas \Set{X \in \Comp}{X \cap \mathcal{D}_{\mathcal{A}} \neq \emptyset}.
\]

\begin{lemma}[Positive definiteness]
The functions $\mathcal{V}$ and $\mathcal{W}$ are positive definite in the sense that
$\mathcal{V}(X),\,\mathcal{W}(X) > 0$  for all $X \in \Comp$ with  $X \cap \mathcal{A} = \emptyset$,
and $\mathcal{V}(X) = \mathcal{W}(X) = 0$ whenever  $X \cap \mathcal{A} \neq \emptyset$.
\end{lemma}

\begin{proof}
The result follows directly from the definitions of $\Psi$, $\mathcal{V}$, and $\mathcal{W}$.
\end{proof}

\begin{lemma}[Monotonicity]\label{lem:Monotonicity}
The functions $\mathcal{V}$ and $\mathcal{W}$ are monotone with respect to set inclusion, i.e., for $ X,Y \in \Comp$,
$X \subseteq Y \;\Rightarrow\; \mathcal{V}(X) \ge \mathcal{V}(Y), 
\quad \mathcal{W}(X) \ge \mathcal{W}(Y)$.

\end{lemma}

\begin{proof}
The result follows immediately from the definitions and the monotonicity of the infimum operator.
\end{proof}

By adapting the proof of Theorem~\ref{thm:DOASublevel}, we obtain the
following stronger result.

\begin{lemma}\label{lem:VFinite}
For every $X\in\Comp$,
\[
X\in\CompDOA
\quad\Longleftrightarrow\quad
\mathcal{V}(X)<\infty.
\]
\end{lemma}

\begin{proof}
Suppose first that $X\in\CompDOA$. By definition, there exists
$x\in X\cap\mathcal{D}_{\mathcal{A}}$. Theorem~\ref{thm:DOASublevel}
therefore implies that
$\mathcal{V}(\{x\})$ is finite. Since $\{x\}\subseteq X$, the monotonicity of $\mathcal{V}$
(Lemma~\ref{lem:Monotonicity}) gives
 $\mathcal{V}(X)\leq \mathcal{V}(\{x\})<\infty$.

Conversely, let $X\in\Comp$ satisfy $\mathcal{V}(X)<\infty$.  By the definition of $\mathcal{V}$ in \eqref{eq:Lyapunov1}, we  have
$\lim_{k\to\infty}\Psi(\mathcal{R}(X,k))=0$. The lower bound on $\alpha$ in \eqref{eq:AlphaBounds} then implies
$\lim_{k\to\infty}
\inf_{y\in\mathcal{R}(X,k)}
\dist(y,\mathcal{A})=0$. Hence, there exists $K\in\mathbb{Z}_{+}$ such that

$\inf_{y\in\mathcal{R}(X,K)}
\dist(y,\mathcal{A})\leq r$.

Because $\mathcal{R}(X,K)$ is compact and the distance function is
continuous, there exists $y\in\mathcal{R}(X,K)$ satisfying
 $\dist(y,\mathcal{A})\leq r$. By the definition of the reachable set, there exist $x\in X$ and
$\pi_{1}\in U^{\mathbb{Z}_{+}}$ such that
$y=\varphi_x^{\pi_{1}}(K)$. Since $y$ lies in the $r$-neighborhood of $\mathcal{A}$, local
$\ell^p$-stabilizability and the concatenation argument used in the
proof of Theorem~\ref{thm:DOASublevel} yield an input signal
$\pi\in U^{\mathbb{Z}_{+}}$ such that $\lim_{k\to\infty}\dist(\varphi_x^\pi(k),\mathcal{A})=0$. Thus, $x\in\mathcal{D}_{\mathcal{A}}$, and therefore
$X\in\CompDOA$.
\end{proof}
We next study continuity properties of $\mathcal{V}$. The following technical result will be used, whose proof is a trivial adaptation of the proof of Lemma 15 in \cite{serry2026safe}.

\begin{lemma}\label{Lem:InfIsCts}
Let $g:\mathbb{R}^n \to \mathbb{R}$ be continuous, and define $G: D \subseteq \mathcal{K}(\mathbb{R}^n) \to \mathbb{R}$ by $G(X) \defas \inf_{x \in X} g(x), \quad X \in D$. Then $G$ is continuous on $D$ with respect to the Hausdorff distance $\dh$.
\end{lemma}

\begin{theorem}[Continuity of $\mathcal{V}$]\label{thm:VCts}
The function $\mathcal{V}$ is continuous on $\CompDOA$.

\end{theorem}

\begin{proof}
Recall the constants $M$, $r$, and the function $\lambda$ from
Assumption~\ref{Assumptions}, as well as the definition of $\Psi$ in
\eqref{eq:Psi}. Let $X\in\CompDOA$, and let $\varepsilon>0$ be arbitrary.
Without loss of generality, suppose that $\varepsilon\leq r$.

Since $X\cap\mathcal{D}_{\mathcal{A}}\neq\emptyset$, there exist
$x\in X\cap\mathcal{D}_{\mathcal{A}}$ and an input signal
$\pi_1\in U^{\mathbb{Z}_{+}}$ such that
\[
\lim_{k\to\infty}
\dist\bigl(\varphi_x^{\pi_1}(k),\mathcal{A}\bigr)=0.
\]
Consequently, there exists $K\in\mathbb{Z}_{+}$ for which
\[
y:=\varphi_x^{\pi_1}(K)\in\mathcal{R}(X,K)
\]
satisfies
\[
\dist(y,\mathcal{A})\leq\frac{\varepsilon}{2}\leq r.
\]
In particular,
\[
\inf_{z\in\mathcal{R}(X,K)}\dist(z,\mathcal{A})
\leq\frac{\varepsilon}{2}.
\]
By $\ell_p$-stabilizability, there exists
$\pi_2\in U^{\mathbb{Z}_{+}}$ such that
\[
\dist\bigl(\varphi_y^{\pi_2}(k),\mathcal{A}\bigr)
\leq M\lambda(\varepsilon,k),
\qquad k\in\mathbb{Z}_{+}.
\]
Let $\pi\in U^{\mathbb{Z}_{+}}$ be the concatenation of $\pi_1$ and
$\pi_2$, as defined in \eqref{eq:Conctenation}. Then
\[
\dist\bigl(\varphi_x^\pi(K+k),\mathcal{A}\bigr)
\leq M\lambda(\varepsilon,k),
\qquad k\in\mathbb{Z}_{+}.
\]
Since $\varphi_x^\pi(K+k)\in\mathcal{R}(X,K+k)$, it follows that
\[
\inf_{z\in\mathcal{R}(X,K+k)}\dist(z,\mathcal{A})
\leq M\lambda(\varepsilon,k),
\qquad k\in\mathbb{Z}_{+}.
\]

Next, since $\mathcal{D}_{\mathcal{A}}$ is open and
$X\cap\mathcal{D}_{\mathcal{A}}\neq\emptyset$, there exists
$\delta_1>0$ such that
\[
Y\in\Comp,\quad \dh(X,Y)\leq\delta_1
\quad\Rightarrow\quad
Y\in\CompDOA.
\]
Moreover, by the continuity of the reachable-set map
$\mathcal{R}(\cdot,k)$ for $k\in\intcc{0;K}$
(Lemma~\ref{lem:SemiGroup}), the continuity of $\alpha$, and
Lemma~\ref{Lem:InfIsCts}, there exists $\delta_2>0$ such that, for every
$Y\in\CompDOA$ satisfying $\dh(X,Y)\leq\delta_2$,
\begin{align}
\abs{
\sum_{k=0}^{K-1}\Psi\bigl(\mathcal{R}(X,k)\bigr)
-
\sum_{k=0}^{K-1}\Psi\bigl(\mathcal{R}(Y,k)\bigr)
}
&\leq\varepsilon, \label{eq:finite_part}\\
\dh\bigl(\mathcal{R}(X,K),\mathcal{R}(Y,K)\bigr)
&\leq\frac{\varepsilon}{2}. \label{eq:R_cont}
\end{align}
Set $\delta:=\min\{\delta_1,\delta_2\}$, and let
$Y\in\Comp$ satisfy $\dh(X,Y)\leq\delta$. 
 By the property $\inf_{z\in\mathcal{R}(X,K)}\dist(z,\mathcal{A})
\leq {\varepsilon}/{2}$, and the compactness of $\mathcal{R}(X,K)$, there exists
$z\in\mathcal{R}(X,K)$ such that
$\dist(z,\mathcal{A})\leq {\varepsilon}/{2}$. Since $\mathcal{A}$ is compact, there exists $a\in\mathcal{A}$ such that
$\norm{z-a}=\dist(z,\mathcal{A})\leq{\varepsilon}/{2}$.
Condition~\eqref{eq:R_cont} and the compactness of $\mathcal{R}(Y,K)$ implies the existence of
$\tilde{y}\in\mathcal{R}(Y,K)$ such that
$\norm{\tilde{y}-z}\leq {\varepsilon}/{2}$.
Therefore, by the triangle inequality,
$\norm{\tilde{y}-a}
\leq \norm{\tilde{y}-z}+\norm{z-a}
\leq \varepsilon$. For this particular $\tilde{y}$, we have
$\dist(\tilde{y},\mathcal{A})\leq\norm{\tilde{y}-a}\leq \varepsilon$, and hence
\[
\inf_{y\in\mathcal{R}(Y,K)}\dist(y,\mathcal{A})
\leq \varepsilon\leq r.
\]
 Applying \(\ell_p\)-stabilizability to \(\widetilde y\), together with the semigroup property of reachable sets, yields

\[
\inf_{z\in\mathcal{R}(Y,K+k)}\dist(z,\mathcal{A})
\leq M\lambda(\varepsilon,k),
\qquad k\in\mathbb{Z}_{+}.
\]
Therefore, by the definition of $\Psi$ and the upper bound on $\alpha$,
\begin{align*}
\sum_{k=K}^{\infty}\Psi\bigl(\mathcal{R}(Y,k)\bigr)
&\leq
\overline{\alpha}
\sum_{k=0}^{\infty}
\inf_{z\in\mathcal{R}(Y,K+k)}
\dist(z,\mathcal{A})^{\bar p}\\
&\leq
\overline{\alpha}M^{\bar p}
\sum_{k=0}^{\infty}\lambda(\varepsilon,k)^{\bar p}.
\end{align*}
The same estimate holds for the corresponding tail of
$\mathcal{V}(X)$. Hence, together with \eqref{eq:finite_part},
\[
\abs{\mathcal{V}(X)-\mathcal{V}(Y)}
\leq
\varepsilon
+
2\overline{\alpha}M^{\bar p}
\sum_{k=0}^{\infty}\lambda(\varepsilon,k)^{\bar p}.
\]
Because $\varepsilon>0$ is arbitrary, and because the function $s\longmapsto \sum_{k=0}^{\infty}\lambda(s,k)^{\bar p}$ is continuous at $s=0$ by the properties of $\lambda$ and the
Weierstrass $M$-test, with value zero at $s=0$, the right-hand side can
be made arbitrarily small. Thus, $\mathcal{V}$ is continuous at $X$. Since
$X\in\CompDOA$ was arbitrary, $\mathcal{V}$ is continuous on
$\CompDOA$.
\end{proof}

Next, we want to illustrate the continuity of $\mathcal{W}$ over $\Comp$. This will require analyzing the asymptotic behavior of $\mathcal{V}$ near its singular values.
\begin{theorem}[Blowup of $\mathcal{V}$]\label{thm:DOABoundary}
Let $X \in \Comp$ satisfy $X \cap \mathcal{D}_{\mathcal{A}} = \emptyset$, and let $\{X_k\}_{k\in \mathbb{Z}_{+}} \subset \Comp$ be a sequence converging to $X$ in the Hausdorff distance. Then
$\lim_{k\to\infty} \mathcal{V}(X_k) = \infty$.

\end{theorem}
\begin{proof}
By Lemma~\ref{lem:VFinite}, for $Y\in \Comp$, $V(Y)<\infty\Leftrightarrow
Y\in \CompDOA$.
Consequently, $V(Y)=\infty$ whenever
$Y\notin\CompDOA$. It therefore suffices to consider an arbitrary subsequence of
$\{X_k\}_{k\in\mathbb Z_+}$ contained in
$\mathcal K_{\mathcal D_{\mathcal A}}(\mathbb R^n)$.
Relabeling this subsequence, we may assume without loss of generality
that $X_k \in \CompDOA$ for all $k \in \mathbb{Z}_{+}$. 

Recall the constants $M$, $r$, and function $\lambda$ from Assumption~\ref{Assumptions}. Fix $\theta \in \intoo{0,r/M}$. For each $k \in \mathbb{Z}_{+}$, define $T_k \in \mathbb{Z}_{+}$ as the smallest integer such that
$$\mathcal{R}(X_k,T_k) \cap \big(\mathcal{A} + \theta \mathbb{B}_n\big) \neq \emptyset.
$$
By the minimality of $T_k$, for all $j \in \intcc{0;T_k-1}$ (assuming $T_{k}\geq 1$),
$$
\inf_{y \in \mathcal{R}(X_k,j)} \dist(y,\mathcal{A}) \ge \theta.
$$
\noindent\textbf{Case 1: $T_k \to \infty$.}
In this case, and using the definitions of $\mathcal{V}$ and $\Psi$ in \eqref{eq:Lyapunov1} and \eqref{eq:Psi}, respectively, and the lower bound on $\alpha$ in \eqref{eq:AlphaBounds}, $$\mathcal{V}(X_k)
\ge \sum_{j=0}^{T_k-1} \Psi\big(\mathcal{R}(X_k,j)\big)
\ge \underline{\alpha}\,\theta^{\bar p}\,T_k,
$$
which implies $\mathcal{V}(X_k)\to\infty$.

\noindent\textbf{Case 2:} $T_k\not\rightarrow \infty$.
Then there exists a subsequence, again denoted $\{T_{k}\}$, and $T \in \mathbb{Z}_{+}$ such that $T_k \le T$ for all $k$. Let $k\in \mathbb{Z}_{+}$. Since $\theta<r/M \le r$, we have
$\mathcal{R}(X_k,T_k) \cap \big(\mathcal{A} + r \mathbb{B}_n\big) \neq \emptyset$. By $\ell_p$-stabilizability, it follows that, for $ \quad (k,j) \in \mathbb{Z}_{+}\times \mathbb{Z}_{+}$, $$\inf_{y \in \mathcal{R}(X_k,T_k + j)} \dist(y,\mathcal{A})
\le M \lambda(\theta,j)\leq M \theta. 
$$
Therefore,
$$
\inf_{y \in \mathcal{R}(X_k,T)} \dist(y,\mathcal{A}) \le M\theta, \quad k \in \mathbb{Z}_{+}.
$$
Using continuity of the reachable set map with respect to the Hausdorff distance $\dh$, the continuity of the distance function $\dist$,  and Lemma \ref{Lem:InfIsCts}, we pass to the limit to obtain
$$\inf_{y \in \mathcal{R}(X,T)} \dist(y,\mathcal{A}) \le M\theta \le r.
$$
Applying $\ell_p$-stabilizability again yields
$$
\inf_{y \in \mathcal{R}(X,T+j)} \dist(y,\mathcal{A})
\le M \lambda(r,j), \quad j \in \mathbb{Z}_{+}.
$$
Hence,
$\lim_{j\to\infty} \inf_{y \in \mathcal{R}(X,j)} \dist(y,\mathcal{A}) = 0$, which implies that $X \cap \mathcal{D}_{\mathcal{A}} \neq \emptyset$, i.e., $X \in \CompDOA$. This contradicts the assumption $X \cap \mathcal{D}_{\mathcal{A}} = \emptyset$. Therefore, only Case~1 is possible, and $\mathcal{V}(X_k) \to \infty$.
\end{proof}

Using Theorems \ref{thm:VCts} and \ref{thm:DOABoundary}, we have:
\begin{theorem}[Continuity of $\mathcal{W}$]
   $\mathcal{W}$ is continuous over  $\Comp$. 
\end{theorem}

\subsection{Characterizing the value functions: Bellman-type equations}
\label{sec:Zubov}
In this section, we derive Bellman-type equations corresponding to the functions $\mathcal{V}$ and $\mathcal{W}$, respectively.

\begin{theorem}\label{Thm:LyapunovEquation}
 $\mathcal{V}$ satisfies the  equation (w.r.t. the function $v$)
\begin{align}\label{eq:LyapunovEqn}
v(X)
= \Psi(X)+v(F(X)),~X\in \Comp.
\end{align}
\end{theorem}

\begin{proof}
See the proof of Theorem 18 in \cite{serry2026safe}.
\end{proof}
\begin{theorem}\label{thm:ZubovEqns} For $X\in \Comp$,
 $\mathcal{W}$ satisfies the  equations (w.r.t. the function $w$)
\begin{equation}\label{eq:ZubovEqn1}
w(X)-w({F}(X))=\xi(X)(1-w(F(X))),
\end{equation}
and
\begin{equation}\label{eq:ZubovEqn2}
w(X)-w({F}(X))=\beta(X)(1-w(X)),
\end{equation}
where 
\begin{equation}\label{eq:xi}
\xi(X)\defas 1-\exp(-\Psi(X)),
\end{equation}
and
\begin{equation}\label{eq:beta}
\beta(X)\defas \exp(\Psi(X))-1.
\end{equation}
\end{theorem}

\begin{proof}
See the proofs of Theorems 19 and 20 in \cite{serry2026safe}.
\end{proof}
We have shown that the value functions $\mathcal{V}$ and $\mathcal{W}$ are solutions to the equations \eqref{eq:LyapunovEqn} and \eqref{eq:ZubovEqn1}, respectively. Next, we show that the solutions to these equations are unique with respect to functions that satisfy a particular limit condition. We start with the following technical result:
\begin{lemma}\label{lem:boundedness}
    Assume that $w\colon \Comp\rightarrow \mathbb{R}$   satisfies equation \eqref{eq:ZubovEqn2} over $\Comp$, is finite over $\CompDOA$, and, for any  $X\in \CompDOA$,  $\lim_{k\rightarrow \infty} w(\mathcal{R}(X,k))=0
    $. Then, $w(X)<1$ for all  $X\in\CompDOA$.
\end{lemma}
\begin{proof}
Let  $X\in \CompDOA$. Then,  it follows that
 $
     \lim_{k\rightarrow \infty} \inf_{y\in \mathcal{R}(X,k)}\dist(y,\mathcal{A})=0,
 $
 and  $\mathcal{R}(X,k)\in \CompDOA$ for all $k\in \mathbb{Z}_{+}$. 
By contradiction, assume that $w(X)\geq 1$.  Using equation \eqref{eq:ZubovEqn2}, we have 
  $$
w(X)-w(\mathcal{R}(X,1))=\beta(X)(1-w(X))\leq 0.
$$
Hence, $
w(\mathcal{R}(X,1))\geq w(X)\geq 1,$
and by induction, we have $$
w(\mathcal{R}(X,k+1))\geq w(\mathcal{R}(X,k))\geq 1
$$
for all $k\in \mathbb{Z}_{+}$. Hence, $
\lim_{k\rightarrow\infty}w(\mathcal{R}(X,k))\neq 0,
$
which yields a contradiction as $\lim_{k\rightarrow \infty}w(\mathcal{R}(X,k))$ 
must be zero due to the assumption on $w$, and that completes the proof.
\end{proof}

Now we show the uniqueness result for the Bellman-type equation associated with the function $\mathcal{V}$.
\begin{theorem}\label{thm:LyapunovEqnUniqueness}
    Let $\mathbf{v}\colon \Comp\rightarrow \mathbb{R}$   satisfy equation \eqref{eq:LyapunovEqn} over   
    $\Comp$  and be finite over $\CompDOA$. Moreover, for any   $X\in \CompDOA$,  $\lim_{k\rightarrow \infty} \mathbf{v}(\mathcal{R}(X,k))=0
    $. Then,
    $
\mathbf{v}(X)=\mathcal{V}(X)$ for all such $X.
    $
\end{theorem}
\begin{proof}
 Let $X\in \CompDOA$, then $\mathcal{R}(X,k)\in \CompDOA$ for all $k\in \mathbb{Z}_{+}$ and $\mathbf{v}$ is finite over all the reachable set values.  We consequently have, using \eqref{eq:LyapunovEqn}, 
 \begin{align*}
 \mathbf{v}(X)-\mathbf{v}(\mathcal{R}(X,k)) &=\sum_{j=0}^{k-1}(\mathbf{v}(\mathcal{R}(X,j))-\mathbf{v}(\mathcal{R}(X,j+1)))\\
 &=\sum_{j=0}^{k-1}\Psi(\mathcal{R}(X,j)).
  \end{align*}

It follows by assumption that $\lim_{k\rightarrow \infty}\mathbf{v}(\mathcal{R}(X,k))=0$. We also have  $\mathcal{V}(X)<\infty$ (Lemma \ref{lem:VFinite}), then taking the limit as $k\rightarrow \infty$ on the both sides of the above equation results in
    $
     \mathbf{v}(X)=\sum_{j=0}^{\infty}\Psi(\mathcal{R}(X,j))=\mathcal{V}(X).
    $
\end{proof}

Now, we introduce the uniqueness result for the equation associated with the value function $\mathcal{W}$.
\begin{theorem}
    Let $\mathbf{w}\colon \Comp\rightarrow \mathbb{R}$ be a bounded function  satisfying 
    equation \eqref{eq:ZubovEqn1} over $\Comp$ and for any  compact set $X\in \CompDOA$,    $\lim_{k\rightarrow \infty} \mathbf{w}(\mathcal{R}(X,k))=0
    $.  
Then, $\mathbf{w}=\mathcal{W}$.
\end{theorem}
\begin{proof}
   First note that the difference of two solutions, $\mathbf{w}_{1}$ and $\mathbf{w}_{2}$, to equation \eqref{eq:ZubovEqn1} satisfies, for $X\in \Comp$,
\begin{equation}\label{eq:difference1}
    \mathbf{w}_{1}(X)-\mathbf{w}_{2}(X)=(1-\xi(X))(\mathbf{w}_{1}(F(X))-\mathbf{w}_{2}(F(X))).
\end{equation}

\noindent(\textbf{Case 1}) On $\CompDOA$, the function
$\mathbf{v}(\cdot)\defas -\ln(1-\mathbf{w}(\cdot))$ is well defined by
Lemma~\ref{lem:boundedness} and satisfies \eqref{eq:LyapunovEqn}.
Therefore, Theorem~\ref{thm:LyapunovEqnUniqueness} yields
$\mathbf{v}(X)=\mathcal{V}(X)$ for all $X\in\CompDOA$. Consequently,
$\mathbf{w}(X)=\mathcal{W}(X)$ for every $X\in\CompDOA$.

\noindent(\textbf{Case 2}) Let $X\in\Comp$ satisfy
$\mathcal{R}(X,j)\cap\mathcal{D}_{\mathcal{A}}\neq\emptyset$ for some
$j\in\mathbb{N}$. Define
$\Delta_w(\cdot)\defas\mathbf{w}(\cdot)-\mathcal{W}(\cdot)$.
By \eqref{eq:difference1},
\[
\Delta_w(\mathcal{R}(X,j-1))
=
\bigl(1-\xi(\mathcal{R}(X,j-1))\bigr)
\Delta_w(\mathcal{R}(X,j))
=0,
\]
where the last equality follows from Case~1, since
$\mathcal{R}(X,j)\in\CompDOA$. Repeating this argument inductively gives
$\Delta_w(\mathcal{R}(X,k))=0$ for all $k\in\intcc{0;j}$.
In particular, $\mathbf{w}(X)=\mathcal{W}(X)$.

 \noindent(\textbf{Case 3}) Finally, let $X\in\Comp$ be such that
\[
\mathcal{R}(X,k)\cap\mathcal{D}_{\mathcal{A}}=\emptyset,
\qquad \forall k\in\mathbb{Z}_{+}.
\]
By $\ell_p$-stabilizability, Assumption~\ref{Assumptions} guarantees that
\[
\mathcal{A}+r\mathbb{B}\subseteq\mathcal{D}_{\mathcal{A}},
\]
where $r>0$ is the constant specified in Assumption~\ref{Assumptions}.
Therefore,
\[
\inf_{y\in\mathcal{R}(X,k)}\dist(y,\mathcal{A})\geq r,
\qquad \forall k\in\mathbb{Z}_{+}.
\]
Set $\theta\defas r$, and suppose, for contradiction, that
$\mathbf{w}(X)\neq\mathcal{W}(X)$. Again, let
$\Delta_w(\cdot)\defas\mathbf{w}(\cdot)-\mathcal{W}(\cdot)$.
For any two solutions $\mathbf{w}_1$ and $\mathbf{w}_2$ of
\eqref{eq:ZubovEqn1} on $\Comp$, which also solve
\eqref{eq:ZubovEqn2}, we have, for all $Y\in\Comp$, 
\[
\mathbf{w}_1(F(Y))-\mathbf{w}_2(F(Y))
=
\bigl(1+\beta(Y)\bigr)
\bigl(\mathbf{w}_1(Y)-\mathbf{w}_2(Y)\bigr).
\]
Consequently, repeated application of this identity yields, by induction,
\[
\Delta_w(\mathcal{R}(X,j))
=
\prod_{k=0}^{j-1}
\bigl(1+\beta(\mathcal{R}(X,k))\bigr)
\Delta_w(X),
\qquad \forall j\in\mathbb{N}.
\]
Moreover, for every $k\in \mathbb{Z}_{+}$, the definition of $\Psi$ in
\eqref{eq:Psi} and the lower bound in \eqref{eq:AlphaBounds} yield
\begin{align*}
1+\beta(\mathcal{R}(X,k))
&=\exp\bigl(\Psi(\mathcal{R}(X,k))\bigr)\\
&\geq
\exp\left(
\underline{\alpha}
\inf_{y\in\mathcal{R}(X,k)}
\dist(y,\mathcal{A})^{\bar p}
\right)\\
&\geq \exp\bigl(\underline{\alpha}\theta^{\bar p}\bigr).
\end{align*}
Hence, for all $j\in \mathbb{N}$, 
\[
\abs{\Delta_w(\mathcal{R}(X,j))}
\geq
\exp\bigl(j\underline{\alpha}\theta^{\bar p}\bigr)
\abs{\Delta_w(X)}.
\]
Since $\Delta_w(X)\neq0$, the right-hand side diverges as
$j\to\infty$. Thus, for every $K>0$, there exists $j\in\mathbb{N}$ such
that, with $Z\defas\mathcal{R}(X,j)$,
\[
\abs{\Delta_w(Z)}>K.
\]
Because $0\leq\mathcal{W}(Z)\leq1$, we obtain
\[
\abs{\mathbf{w}(Z)}
\geq
\abs{\Delta_w(Z)}-\abs{\mathcal{W}(Z)}
>K-1.
\]
This contradicts the boundedness of $\mathbf{w}$ on $\Comp$.
Therefore, $\mathbf{w}(X)=\mathcal{W}(X)$.
\end{proof}
We conclude this section with the following result, showing that the value functions $\mathcal{V}$ and $\mathcal{W}$ are, in fact, control Lyapunov functions under certain sufficient conditions. 
\begin{theorem}\label{thm:radial-sufficient-condition}
Let \(\rho:\mathbb{R}^{n}\to \mathbb{R}_{+}\) be continuous, and let \(\phi:\mathbb{R}_{+}\to \mathbb{R}_{+}\) be continuous and nondecreasing. Besides the aforementioned properties of $\alpha$ (equation \eqref{eq:AlphaBounds}),  assume that it has the form:
\[
\alpha(x)=\phi(\rho(x)), \qquad x\in \mathbb{R}^{n}.
\]
Additionally, assume that there exists a continuous nondecreasing function \(\Gamma:\mathbb{R}_{+}\to \mathbb{R}_{+}\) such that
\[
\inf_{u\in U}\rho(f(x,u))=\Gamma(\rho(x))
\qquad \forall x\in \mathbb{R}^{n}.
\]
 Then, $\mathcal{V}$ satisfies 
$$
\inf_{u\in U} \mathcal{V}(\{f(x,u)\})<\mathcal{V}(\{x\}),\quad  x\in \mathcal{D}_{\mathcal{A}}\setminus \mathcal{A},
$$
and $\mathcal{W}$ satisfies
$$
\inf_{u\in U}\mathcal{W}(\{f(x,u)\})< \mathcal{W}(\{x\}),
\quad  x\in \mathcal{D}_{\mathcal{A}}\setminus \mathcal{A}.
$$

\end{theorem}

\begin{proof}
We prove the result for $\mathcal{V}$. The statement for $\mathcal{W}$
then follows from the strict monotonicity of the map
$s\mapsto 1-e^{-s}$. For convenience, define
\[
v(x)\defas\mathcal{V}(\{x\}),
\qquad x\in\mathbb{R}^n.
\]
Then
$$
v(x)=\sum_{k=0}^{\infty}m_k(x),~x\in \mathbb{R}^{n}
$$
where
$$
m_k(x)\defas\inf_{y\in\mathcal{R}(\{x\},k)}\alpha(y),~k\in \mathbb{Z}_{+},~x\in \mathbb{R}^{n}.
$$

We first show that
\[
m_{k+1}(x)=\inf_{u\in U}m_k(f(x,u)),
\qquad \forall x\in\mathbb{R}^n,\quad \forall k\in\mathbb{Z}_+.
\]
Indeed, by Lemma~\ref{lem:SemiGroup},
\[
\mathcal{R}(\{x\},k+1)
=
\mathcal{R}(F(\{x\}),k)
=
\bigcup_{u\in U}\mathcal{R}(\{f(x,u)\},k).
\]
Therefore,
\begin{align*}
m_{k+1}(x)
&=
\inf_{y\in\mathcal{R}(\{x\},k+1)}\alpha(y)\\
&=
\inf_{u\in U}
\inf_{y\in\mathcal{R}(\{f(x,u)\},k)}\alpha(y)\\
&=
\inf_{u\in U}m_k(f(x,u)).
\end{align*}

We next show that, for every $k\in\mathbb{Z}_+$, there exists a
continuous and nondecreasing function
$\phi_k:\mathbb{R}_{+}\to\mathbb{R}_{+}$ such that
\[
m_k(x)=\phi_k(\rho(x)),
\qquad \forall x\in\mathbb{R}^n.
\]
We proceed by induction on $k$. For $k=0$, since
$\mathcal{R}(\{x\},0)=\{x\}$,
\[
m_0(x)
=
\inf_{y\in\{x\}}\alpha(y)
=
\alpha(x)
=
\phi(\rho(x)).
\]
Thus, the claim holds with $\phi_0=\phi$, which is continuous and
nondecreasing by assumption.

Suppose that, for some $k\in\mathbb{Z}_+$, there exists a continuous
and nondecreasing function $\phi_k:\mathbb{R}_{+}\to\mathbb{R}_{+}$ such that
\[
m_k(x)=\phi_k(\rho(x)),
\qquad \forall x\in\mathbb{R}^n.
\]
Then
\[
m_{k+1}(x)
=\inf_{u\in U}m_{k}(f(x,u))=
\inf_{u\in U}\phi_k(\rho(f(x,u))).
\]
Since $U$ is compact and $u\mapsto\rho(f(x,u))$ is continuous, there
exists $u_x\in U$ such that
\[
\rho(f(x,u_x))
=
\min_{u\in U}\rho(f(x,u)).
\]
Because $\phi_k$ is nondecreasing,
\begin{align*}
m_{k+1}(x)
&=
\inf_{u\in U}\phi_k(\rho(f(x,u)))=
\phi_k\!\left(\inf_{u\in U}\rho(f(x,u))\right)\\
&=\phi_k(\rho(f(x,u_x))).
\end{align*}
Note that the minimizer $u_{x}$ is independent of $k$. By the assumption on $\Gamma$,
\[
m_{k+1}(x)=\phi_k(\Gamma(\rho(x))).
\]
Define
\[
\phi_{k+1}(r)\defas\phi_k(\Gamma(r)),
\qquad r\in[0,\infty).
\]
Since both $\phi_k$ and $\Gamma$ are continuous and nondecreasing,
$\phi_{k+1}$ is also continuous and nondecreasing, and this completes the induction.

As shown above, for every $k\in\mathbb{Z}_+$, the input $u_x$ satisfies $m_k(f(x,u_x))
=
\inf_{u\in U}m_k(f(x,u))$.
Consequently, for all $x\in \mathcal{D}_{\mathcal{A}}$, 
\begin{align*}
\inf_{u\in U}\sum_{k=0}^{\infty}m_k(f(x,u))&=\inf_{u\in U}\sum_{k=0}^{\infty}m_k(f(x,u))\\
&=\sum_{k=0}^{\infty}m_k(f(x,u_x))\\
&=
\sum_{k=0}^{\infty}\inf_{u\in U}m_k(f(x,u))\\
&=
\sum_{k=0}^{\infty}m_{k+1}(x)\\
&=
v(x)-\alpha(x)<\infty.
\end{align*}
Finally, for $x\in\mathcal{D}_{\mathcal{A}}\setminus\mathcal{A}$, and using the positive-definiteness of $\alpha$, 
$$
\inf_{u\in U}v(f(x,u))=v(x)-\alpha(x)
<
v(x)
<
\infty.
$$
\end{proof}

\section{DOS Estimation}\label{sec:VFApproximation}

In the previous sections, we introduced value functions defined on metric spaces of compact sets that characterize the DOS. We now leverage these value functions to obtain constructive and practically implementable estimates of the DOS via NN learning.

\subsection{Finite-Horizon Approximation of the Value Functions}
\label{sec:ReachSetApproximation}

Within the proposed learning framework, it is necessary to evaluate the value functions $\mathcal{V}$ and $\mathcal{W}$, at least approximately, over singleton sets. Let $N_s \in \mathbb{N}$ denote a truncation horizon. We define the finite-horizon approximation
$\mathcal{V}_{N_s}(X) \defas 
\sum_{k=0}^{N_s} \Psi\big(\mathcal{R}(X,k)\big),~ X \in \Comp$. In general, the reachable sets $\mathcal{R}(X,k)$ cannot be computed exactly. We therefore consider approximations $\tilde{\mathcal{R}}(X,k)$ such that
$
\mathcal{R}(X,k) \approx \tilde{\mathcal{R}}(X,k), \quad k\in \intcc{0;N_s}.
$
A variety of reachable-set approximation techniques for discrete-time systems are available in the literature (see, e.g., \cite{kuhn1998rigorously,alamo2005guaranteed,yang2020accurate}). These methods typically involve a trade-off between accuracy and computational complexity. Moreover, the chosen approximation must enable tractable evaluation of the function $\Psi$ in \eqref{eq:Psi}. In our setting, it suffices to approximate reachable sets for singleton initial conditions $X=\{x\}$ with $x\in\mathbb{R}^n$.

In Section~\ref{sec:NumericalExamples}, we adopt a trajectory-based approximation that balances accuracy and computational efficiency. Let $N_{\mathrm{traj}} \in \mathbb{N}$ denote the number of sampled input trajectories, and let $\pi^{(j)}:\mathbb{Z}_+ \to U$, $j\in \intcc{1;N_{\mathrm{traj}}}$, be randomly generated input signals. We define
$\tilde{\mathcal{R}}(\{x\},k)
=
\bigcup_{j=1}^{N_{\mathrm{traj}}}
\left\{ \varphi_{x}^{\pi^{(j)}}(k) \right\}$, i.e., the reachable set at time $k$ is approximated by a finite collection of sampled trajectories. This construction renders the evaluation of $\Psi$ computationally tractable. In practice, increasing \(N_{\rm traj}\) can improve the approximation accuracy. Based on these approximations, we define
$\tilde{\mathcal{V}}_{N_s}(X)
\defas 
\sum_{k=0}^{N_s} \Psi\big(\tilde{\mathcal{R}}(X,k)\big)$, and $\tilde{\mathcal{W}}_{N_s}(X)
\defas 
1 - \exp\!\big(-\tilde{\mathcal{V}}_{N_s}(X)\big)$, $X \in \Comp$.

\subsection{Physics-Informed and Set-Based Neural Network Learning}
\label{sec:PINNLearning}

Since evaluation of the value function $\mathcal{W}$ over singleton sets is sufficient for DOS characterization  (Theorem~\ref{thm:DOASublevel}), we seek to approximate the mapping $x \mapsto \mathcal{W}(\{x\})$ for $x \in \mathbb{R}^n$. The resulting approximation can then be used to estimate the DOS.

We propose a physics-informed learning framework that incorporates the set-based Bellman (Zubov-type) equation~\eqref{eq:ZubovEqn1}. A key challenge in embedding \eqref{eq:ZubovEqn1} into the training process is the set-valued nature of the term $F(\{x\})$, which arises due to the presence of multiple admissible controls.

To address this, we impose the following structural assumption.

\begin{assumption}\label{assumption:embedding}
Define
\[
\mathcal{S}
:=
\big\{ \{x\} : x \in \mathbb{R}^n \big\}
\;\cup\;
\big\{ F(\{x\}) : x \in \mathbb{R}^n \big\}.
\]

Assume there exists a mapping
$
\mathcal{T} : \mathcal{S} \to \mathbb{R}^L,
$
for some fixed $L \in \mathbb{N}$, such that $\mathcal{T}$ is injective on $\mathcal{S}$. In particular, each singleton set $\{x\}$ and each reachable set $F(\{x\})$ admits a unique finite-dimensional representation.

This assumption is satisfied when $\mathcal{S}$ is represented using a parameterized class of sets with fixed complexity (e.g., polytopes, zonotopes, or fixed-template sublevel sets), as discussed thoroughly in Remark~6 of \cite{serry2026safe}.
\end{assumption}

Let $\tilde{\omega}_{\mathrm{nn},\theta} : \mathbb{R}^L \to \mathbb{R}$ be a feedforward neural network with a fixed architecture, parameterized by weights and biases collected in $\theta$. We aim to determine $\theta$ such that the composition $\tilde{\omega}_{\mathrm{nn},\theta} \circ \mathcal{T}$
approximates the value function over a compact learning domain $\mathbb{X}_{\ell} \subset \mathbb{R}^n$ satisfying $\mathcal{A} \subseteq \mathbb{X}_{\ell}$.

The training problem is formulated as the minimization of the composite loss
$
\mathcal{L}(\theta)
=
\lambda_{\mathrm{d}} \mathcal{L}_{\mathrm{d}}(\theta)
+
\lambda_{\mathrm{pi}} \mathcal{L}_{\mathrm{pi}}(\theta),
$
where $\lambda_{\mathrm{d}}, \lambda_{\mathrm{pi}} > 0$ are weighting coefficients. The data-driven and physics-informed losses are given by
$
\mathcal{L}_{\mathrm{d}}(\theta)
=
\frac{1}{N_{\mathrm{d}}}
\sum_{i=1}^{N_{\mathrm{d}}}
J_{\mathrm{d},\theta}\bigl(\{z_{\mathrm{d}}^{(i)}\}\bigr),
$
$
\mathcal{L}_{\mathrm{pi}}(\theta)
=
\frac{1}{N_{\mathrm{pi}}}
\sum_{i=1}^{N_{\mathrm{pi}}}
J_{\mathrm{pi},\theta}\bigl(\{z_{\mathrm{pi}}^{(i)}\}\bigr).
$ The data-driven residual is defined as
\[
J_{\mathrm{d},\theta}(\{x\})
=
\Big(
\tilde{\mathcal{W}}_{N_s}(\{x\})
-
\tilde{\omega}_{\mathrm{nn},\theta}\bigl(\mathcal{T}(\{x\})\bigr)
\Big)^2,
\]
which penalizes the discrepancy between the network output and the finite-horizon approximation $\tilde{\mathcal{W}}_{N_s}$ of the value function $\mathcal{W}$. The physics-informed residual is defined by
\begin{align*}
J_{\mathrm{pi},\theta}(\{x\})
&=
\Big(
\tilde{\omega}_{\mathrm{nn},\theta}\bigl(\mathcal{T}(\{x\})\bigr)
-
\tilde{\omega}_{\mathrm{nn},\theta}\bigl(\mathcal{T}(F(\{x\}))\bigr) \\
&\hspace{1.2cm}
-
\xi(\{x\})
\bigl(
1
-
\tilde{\omega}_{\mathrm{nn},\theta}\bigl(\mathcal{T}(F(\{x\}))\bigr)
\bigr)
\Big)^2,
\end{align*}
which enforces the Bellman-type functional equation~\eqref{eq:ZubovEqn1} in a residual minimization sense.

The collocation points 
$\{z_{\mathrm{d}}^{(i)}\}_{i=1}^{N_{\mathrm{d}}} \subset \mathbb{X}_{\ell}$ 
and 
$\{z_{\mathrm{pi}}^{(i)}\}_{i=1}^{N_{\mathrm{pi}}} \subset \mathbb{X}_{\ell}$ 
are independently sampled from $\mathbb{X}_{\ell}$. The former are used to fit the approximate value function, while the latter enforce the governing functional equation. Finally, we define the state-space function $\omega_{\mathrm{nn},\theta}(x)
\defas
\tilde{\omega}_{\mathrm{nn},\theta}\bigl(\mathcal{T}(\{x\})\bigr),
~ x \in \mathbb{R}^n$. The function $\omega_{\mathrm{nn}}$ (index $\theta$ is dropped for convenience) is thus obtained as the composition of the embedding $\mathcal{T}$ with the trained neural network. It serves as a NN value function used for  DOS estimation. In Section \ref{sec:NumericalExamples}, we illustrate the extraction of the NN DOS estimation and stabilizing controllers.

\section{Numerical Examples}
\label{sec:NumericalExamples}
In this section, we illustrate the effectiveness of the proposed framework through two academic examples. In both cases, the CIS is given by $\mathcal{A}=\{\mathbf{0}\}$, and we consider a domain $\mathbb{X}$ for estimation and learning such that $\mathcal{A}\subseteq \mathbb{X}$. In both cases, the origin is an unstable equilibrium for the uncontrolled system ($u=\mathbf{0}$), while the systems are locally stabilizable under suitable feedback. For both examples, the values $F(\{x\})$, for $x \in \mathbb{R}^{n}$, are hyper-rectangles, so we use the hyper-interval-based embedding in Remark 6 in \cite{serry2026safe}. 

All computations are carried out in \textsc{Matlab} on a 13th Gen Intel(R) Core(TM) i7-1355U (1.70~GHz) laptop, with the help of the reachability toolbox CORA~\cite{althoff2015introduction}. For both examples, the learned NNs employ two hidden layers. Each hidden layer contains 20 neurons for the 2D example and 30 neurons for the 3D example. The training parameters are chosen identically for both cases as follows: $N_{s} = 30$ and $N_{\mathrm{traj}} = 5000$, which are used for reachable set approximations; $N_{d} = 5000$ and $N_{\mathrm{pi}} = 10000$, representing the number of sampled points and collocation points used during training, respectively. The loss function weights are set to $\lambda_{\mathrm{d}} = 0.1$ and $\lambda_{\mathrm{pi}} = 1$.

\subsection{DOS Estimation and Stabilizing Control}
For each example, we first design a stabilizing linear state-feedback controller of the form $\Pi_{\mathrm{linear}}(x)=Kx$. Using quadratic Lyapunov analysis, we compute an ellipsoidal region of attraction of the closed-loop system, given by the sublevel set $\mathbb{E}_{c_1}=\{x\in \mathbb{X} : \nu(x)\leq c_{1}\}$, where $\nu(x)=x^{\top}Px$ and $P\succ 0$. The controller $\Pi_{\mathrm{linear}}$ is chosen such that the input constraints are satisfied within this region, i.e., $x\in \mathbb{E}_{c_{1}} \Rightarrow \Pi_{\mathrm{linear}}(x)\in U$. Note that $\mathbb{E}_{c_{1}}\subseteq \mathcal{D}_{\mathcal{A}}$.

Next, we employ a (non-formal) grid-based enlargement procedure to expand this initial estimate and obtain a larger ellipsoidal region of stabilization $\mathbb{E}_{c_2}=\{x\in \mathbb{X} : \nu(x)\leq c_{2}\}$, with $c_2 \geq c_1$, such that for $x\in \mathbb{X}$ satisfying $c_{1}\leq \nu(x)\leq c_{2}$, there exists $u\in U$ for which $\nu(f(x,u)) - \nu(x) < 0$ and $f(x,u)\in \mathbb{X}$.

Subsequently, we learn a neural network approximation $\omega_{\mathrm{nn}}$ of the value function over $\mathbb{X}$ using the physics-informed framework described in Section~\ref{sec:PINNLearning}. When appropriate, we set $\alpha = c\,\nu$ for some $c>0$, so that the sublevel sets of $\omega_{\mathrm{nn}}$ near the origin coincide with those of $\nu$. Since $\nu$ provides local stabilization guarantees, this choice enforces consistency near the origin while enabling larger estimates away from it.

A direct approximation of the DOS can then be defined as $\mathbb{W}_{r}=\{x\in \mathbb{X} : \omega_{\mathrm{nn}}(x)\leq r\}, ~ r<1$. While such a set may provide a good approximation of the DOS when $\omega_{\mathrm{nn}}$ is accurate, it does not, in general, yield a constructive stabilizing policy.

Although the exact value function $\mathcal{W}$ is guaranteed to be a control Lyapunov function only under the sufficient conditions of Theorem~\ref{thm:radial-sufficient-condition}, it nevertheless serves as a strong candidate for stabilization. Motivated by this, we compute, via grid-based search, two scalars $\omega_{1}$ and $\omega_{2}$ with $\omega_{1}<\omega_{2}$ such that: (i) for  $x\in \mathbb{X}$ with $\omega_{\mathrm{nn}}(x)\leq \omega_{1}$, we have $\nu(x)\leq c_{2}$; and (ii) for  $x\in \mathbb{X}$ with $\omega_{1}\leq \omega_{\mathrm{nn}}(x)\leq \omega_{2}$, there exists $u\in U$ such that $\omega_{\mathrm{nn}}(f(x,u)) - \omega_{\mathrm{nn}}(x) < 0$ and $f(x,u)\in \mathbb{X}$. We then define the neural network estimate of the DOS as $\mathcal{D}_{\mathrm{nn}} = \{x\in \mathbb{X} : \omega_{\mathrm{nn}}(x)\leq \omega_{2}\}$.

Finally, we construct a stabilizing controller $\Pi:\mathcal{D}_{\mathrm{nn}}\to U$ satisfying:
\[
\Pi(x)\in
\begin{cases}
\{Kx\}, & {\omega_{\mathrm{nn}}(x)\leq \omega_{1}\wedge  \nu(x)\leq c_{1}}, \\[4pt]
\mathcal{U}_{\nu}(x), 
& {\omega_{\mathrm{nn}}(x)\leq \omega_{1}\wedge  c_{1}<\nu(x)\leq c_{2}}, \\[6pt]
\mathcal{U}_{\omega}(x), 
&{ \omega_{1}<\omega_{\mathrm{nn}}(x)\leq \omega_{2}},
\end{cases}
\]
where $\mathcal{U}_{\nu}(x)\defas \{u\in U \mid \nu(f(x,u))<\nu(x), f(x,u)\in \mathbb{X}\}$, and $\mathcal{U}_{\omega}(x)\defas \{u\in U \mid \omega_{\mathrm{nn}}(f(x,u))<\omega_{\mathrm{nn}}(x), f(x,u)\in \mathbb{X}\}$.

For each example, we plot the largest ellipsoidal estimate $\mathbb{E}_{c_{2}}$, and the neural network estimate $\mathcal{D}_{\mathrm{nn}}$. In addition, we generate closed-loop trajectories with initial conditions in $\mathcal{D}_{\mathrm{nn}}$ and plot their state  2-norms as functions of time and associated control input values.

\subsection{System Equations}

We consider two benchmark systems. The first system is two-dimensional and is given by
\begin{equation}\label{eq:2Dexample}
f(x,u)=
\begin{pmatrix}
x_{1} + 0.1 x_{2} \\
x_{2} + 0.1 \left( x_{1} + x_{1}^3 + x_{2} + u \right)
\end{pmatrix},
\end{equation}
with input set $U=\intcc{-0.5,0.5}$ and learning/estimation domain $\mathbb{X}=\intcc{-1,1}^2$. The CPU time for data generation and  learning  is approximately 5 minutes.

The second system is three-dimensional and is given by
\begin{equation}\label{eq:3Dexample}
f(x,u)=
\begin{pmatrix}
x_{1} + 0.1 x_{2} \\
x_{2} + 0.1 x_{3} \\
x_{3} + 0.1 \left( x_{1} + x_{1}^3 - x_{2} + u \right)
\end{pmatrix},
\end{equation}
with input set $U=\intcc{-1,1}$ and learning/estimation domain $\mathbb{X}=\intcc{-1,1}^3$. For this system, the CPU time for data generation and  learning  is approximately 15 minutes.
\subsection{Results}
\begin{figure*}
    \centering
    \includegraphics[width=0.32\linewidth]{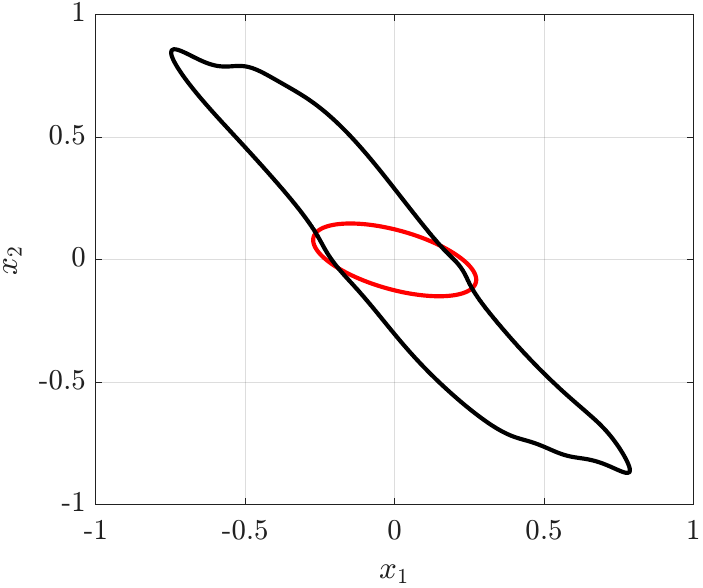}
\includegraphics[width=0.32\linewidth]{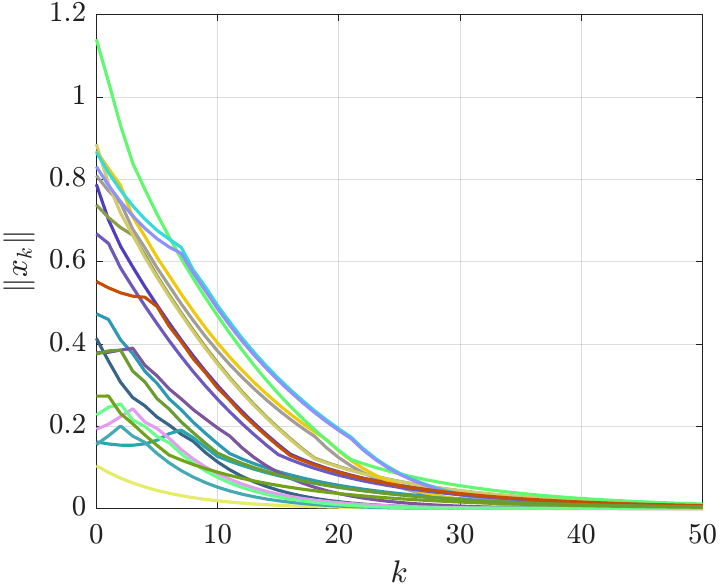}
\includegraphics[width=0.32\linewidth]{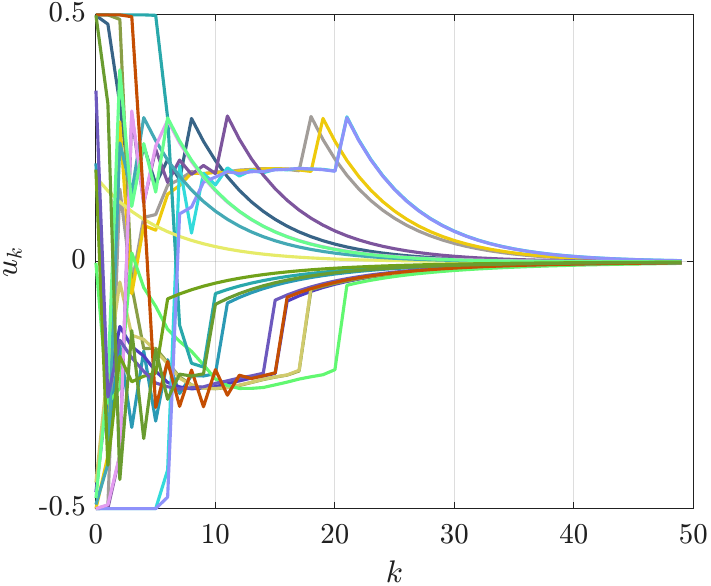}
\caption{Estimated domains of stabilization (left), where the NN-based estimate is shown in black and the ellipsoidal estimate in red; norm profiles of the corresponding closed-loop trajectories initialized in $\mathcal{D}_{\mathrm{nn}}$ under the controller $\Pi$ (middle); and the corresponding control inputs (right) for system \eqref{eq:2Dexample}.}
    \label{fig:2D}
\end{figure*}
\begin{figure*}
    \centering
\includegraphics[width=0.32\linewidth]{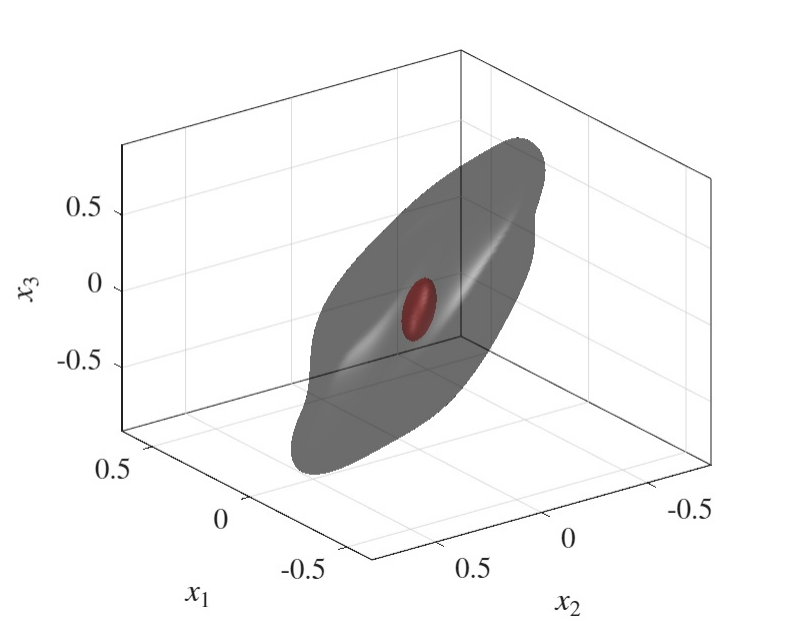}
\includegraphics[width=0.32\linewidth]{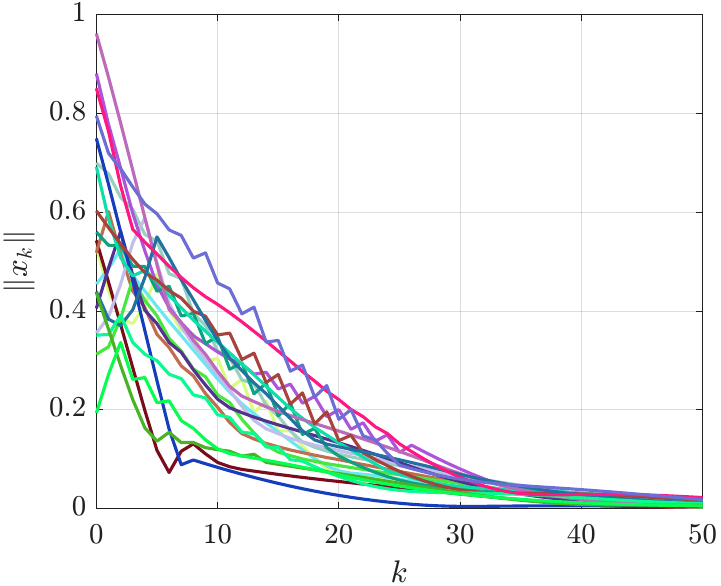}
\includegraphics[width=0.32\linewidth]{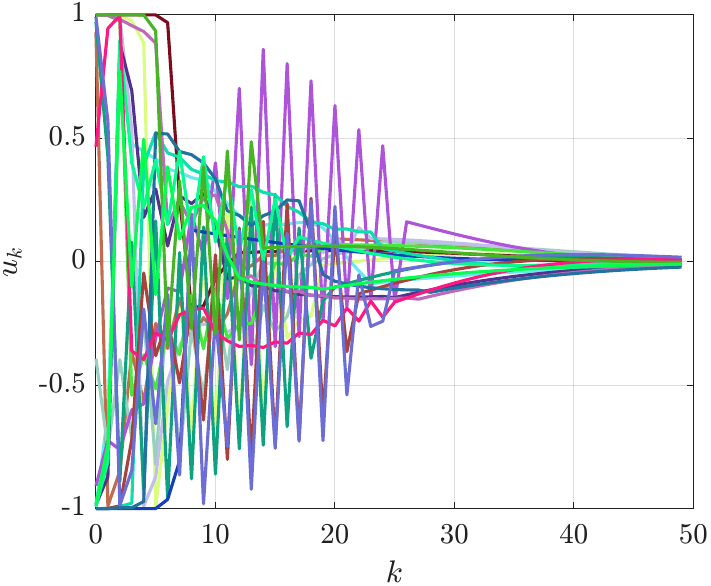}
\caption{Estimated domains of stabilization (left), where the NN-based estimate is shown in gray and the ellipsoidal estimate in red; norm profiles of the corresponding closed-loop trajectories initialized in $\mathcal{D}_{\mathrm{nn}}$ under the controller $\Pi$ (middle); and the corresponding control inputs (right) for system \eqref{eq:3Dexample}.}
    \label{fig:3D}
\end{figure*}
Figures \ref{fig:2D} and \ref{fig:3D} illustrate that the proposed NN framework yields significantly less conservative estimates of the DOS compared to quadratic Lyapunov-based approaches. We note that, in computing the neural estimate $\mathcal{D}_{\mathrm{nn}}$ for both systems, $\omega_{2}\approx 0.97$, i.e., close to $1$, indicating that the estimates closely match the theoretical DOSs, defined as the strict $1$-sublevel sets of the value functions (Theorem \ref{thm:DOASublevel}). Moreover, the figures illustrate the effectiveness of the controller $\Pi$, extracted from the learned value function, in steering system trajectories to the origin while satisfying the input constraints. These results highlight the promise of the proposed framework for control synthesis in more application-relevant settings, which will be the focus of future work.
\section{Conclusion}
\label{sec:Conclusion}
In this paper, we proposed a novel framework for the estimation of DOSs for input-constrained controlled discrete-time systems. The approach is built on value functions defined on metric spaces of compact sets, which provide a characterization of the DOS. We established key properties of these value functions and derived the associated Bellman-type functional equations. Leveraging this characterization, we developed a physics-informed NN training framework to learn the value functions directly from the governing equations. The effectiveness of the proposed methodology in estimating DOSs and synthesizing stabilizing controllers was demonstrated through two numerical examples.

Future work will focus on extending this framework to obtain certified (i.e., formally verified) estimates of the DOS, along with NN-based controllers that admit rigorous correctness guarantees. In particular, we aim to integrate neural network verification tools, such as $\alpha,\beta$-CROWN and SMT-based solvers (e.g., \texttt{dReal}), to certify the learned value functions and the resulting closed-loop behavior. Additionally, we will investigate extensions to systems with state constraints and disturbances. Finally, we will seek to relax the sufficient conditions ensuring that the proposed value functions qualify as control Lyapunov functions, and to derive practically verifiable conditions that can be implemented in computational settings.
\section*{Acknowledgment}

ChatGPT~5.6 \cite{chatgpt2026} was used to assist with proofreading, editorial revisions, and  drafting and checking selected theorem statements and proof steps. The authors take full responsibility for the content of this paper.


\bibliographystyle{ieeetr}

\end{document}